\documentclass{article}%
\usepackage{amssymb}
\usepackage{amsmath}
\usepackage{amsfonts}
\usepackage{graphicx}%
\providecommand{\U}[1]{\protect\rule{.1in}{.1in}}
\newtheorem{theorem}{Theorem}[section]

\newtheorem{conjecture}[theorem]{Conjecture}
\newtheorem{corollary}[theorem]{Corollary}

\newtheorem{lemma}[theorem]{Lemma}

\newtheorem{problem}[theorem]{Problem}
\newtheorem{proposition}[theorem]{Proposition}

\newenvironment{proof}[1][Proof]{\noindent\textbf{#1.} }{\ \rule{0.5em}{0.5em}}
\begin{document}

\title{On K\"{o}nig-Egerv\'{a}ry Collections of Maximum Critical Independent Sets}
\author{Vadim E. Levit\\Department of Computer Science\\Ariel University, Israel\\levitv@ariel.ac.il
\and Eugen Mandrescu\\Department of Computer Science\\Holon Institute of Technology, Israel\\eugen\_m@hit.ac.il}
\date{}
\maketitle

\begin{abstract}
Let $G$ be a simple graph with vertex set $V\left(  G\right)  $. A set
$S\subseteq V\left(  G\right)  $ is \textit{independent} if no two vertices
from $S$ are adjacent. The graph $G$ is known to be a K\"{o}nig-Egerv\'{a}ry
if $\alpha\left(  G\right)  +\mu\left(  G\right)  =$ $\left\vert V\left(
G\right)  \right\vert $, where $\alpha\left(  G\right)  $ denotes the size of
a maximum independent set and $\mu\left(  G\right)  $ is the cardinality of a
maximum matching.

The number $d\left(  X\right)  =$ $\left\vert X\right\vert -\left\vert
N(X)\right\vert $ is the \textit{difference} of $X\subseteq V\left(  G\right)
$, and a set $A\in\mathrm{Ind}(G)$ is \textit{critical} if $d(A)=\max
\{d\left(  I\right)  :I\in\mathrm{Ind}(G)\}$ \cite{Zhang1990}.

Let $\Omega(G)$ denote the family of all maximum independent sets. Let us say
that a family $\Gamma\subseteq\mathrm{Ind}(G)$ is a K\"{o}nig-Egerv\'{a}ry
collection if $\left\vert \bigcup\Gamma\right\vert +\left\vert \bigcap
\Gamma\right\vert =2\alpha(G)$ \cite{JarLevMan2015a}.

In this paper, we show that if the family of all maximum critical independent
sets of a graph $G$ is a K\"{o}nig-Egerv\'{a}ry collection, then $G$\ is a
K\"{o}nig-Egerv\'{a}ry graph. It generalizes one of our conjectures recently
validated in \cite{Short2015}.


\end{abstract}

\section{Introduction}

Throughout this paper $G$ is a finite simple graph with vertex set $V(G)$ and
edge set $E(G)$. If $X\subseteq V\left(  G\right)  $, then $G[X]$ is the
subgraph of $G$ induced by $X$. By $G-W$ we mean either the subgraph
$G[V\left(  G\right)  -W]$, if $W\subseteq V(G)$, or the subgraph obtained by
deleting the edge set $W$, for $W\subseteq E(G)$. In either case, we use
$G-w$, whenever $W$ $=\{w\}$.

The \textit{neighborhood} $N(v)$ of a vertex $v\in V\left(  G\right)  $ is the
set $\{w:w\in V\left(  G\right)  $ \textit{and} $vw\in E\left(  G\right)  \}$.
The \textit{neighborhood} $N(A)$ of $A\subseteq V\left(  G\right)  $ is
$\{v\in V\left(  G\right)  :N(v)\cap A\neq\emptyset\}$, and $N[A]=N(A)\cup A$.

A set $S\subseteq V(G)$ is \textit{independent} if no two vertices from $S$
are adjacent, and by $\mathrm{Ind}(G)$ we mean the family of all the
independent sets of $G$. An independent set of maximum size is a
\textit{maximum independent set} of $G$, and $\alpha(G)=\max\{\left\vert
S\right\vert :S\in\mathrm{Ind}(G)\}$.

Let $\Omega(G)$ denote the family of all maximum independent sets,
\[
\mathrm{core}(G)=%
{\displaystyle\bigcap}
\{S:S\in\Omega(G)\}\text{ \cite{LevMan2002a}, and }\mathrm{corona}(G)=%
{\displaystyle\bigcup}
\{S:S\in\Omega(G)\}\text{ \cite{BorosGolLev}.}%
\]

If $A\in\Omega(G[A])$, then $A$ is called a local maximum independent set of
$G$ \cite{LevMan2002b}.

A \textit{matching} is a set $M$ of pairwise non-incident edges of $G$. A
matching of maximum cardinality, denoted $\mu(G)$, is a \textit{maximum
matching}.

For $X\subseteq V(G)$, the number $\left\vert X\right\vert -\left\vert
N(X)\right\vert $ is the \textit{difference} of $X$, denoted $d(X)$. The
\textit{critical difference} $d(G)$ is $\max\{d(X):X\subseteq V(G)\}$. The
number $\max\{d(I):I\in\mathrm{Ind}(G)\}$ is the \textit{critical independence
difference} of $G$, denoted $id(G)$. Clearly, $d(G)\geq id(G)$. It was shown
in \cite{Zhang1990} that $d(G)$ $=id(G)$ holds for every graph $G$. If $A$ is
an independent set in $G$ with $d\left(  X\right)  =id(G)$, then $A$ is a
\textit{critical independent set} \cite{Zhang1990}.

\begin{theorem}
\label{th6}\cite{NemhTrott1975} Every local maximum independent set is a
subset of a maximum independent set.
\end{theorem}

\begin{proposition}
\label{Prop2}\cite{LevMan2012b} Each critical independent set is a local
maximum independent set.
\end{proposition}

Combining Theorem \ref{th6} and Proposition \ref{Prop2} one can conclude with
the following.

\begin{corollary}
\label{th3}\cite{ButTruk2007} Every critical independent set can be enlarged
to a maximum independent set.
\end{corollary}

For a graph $G$, let us denote
\begin{gather*}
\mathrm{\ker}(G)=%
{\displaystyle\bigcap}
\left\{  A:A\text{ \textit{is a critical independent set}}\right\}  \text{
\cite{LevMan2012a},}\\
\mathrm{MaxCritIndep}(G)=\left\{  S:S\text{ \textit{is a maximum critical
independent set}}\right\}  \text{,}\\
\mathrm{nucleus}(G)=%
{\displaystyle\bigcap}
\mathrm{MaxCritIndep}(G)\text{ \cite{JarLevMan2015a}, and }\mathrm{diadem}(G)=%
{\displaystyle\bigcup}
\mathrm{MaxCritIndep}(G)\text{ \cite{LevMan2014}.}%
\end{gather*}

Clearly, $\mathrm{\ker}(G)\subseteq\mathrm{nucleus}(G)$ and, by Corollary
\ref{th3}, the inclusion $\mathrm{diadem}(G)\subseteq\mathrm{corona}(G)$ is
true for each graph $G$.

\begin{theorem}
\label{th4}\cite{LevMan2012a} For a graph $G$, the following assertions are true:

\emph{(i)} $\mathrm{\ker}(G)\subseteq\mathrm{core}(G)$;

\emph{(ii)} if $A$ and $B$ are critical in $G$, then $A\cup B$ and $A\cap B$
are critical as well.
\end{theorem}

Let us consider the graphs $G_{1}$ and $G_{2}$ from Figure \ref{fig22222}:
$\mathrm{core}(G_{1})=\left\{  a,b,c,d\right\}  $ and it is a critical set,
while $\mathrm{core}(G_{2})=\left\{  x,y,z,w\right\}  $ and it is not
critical.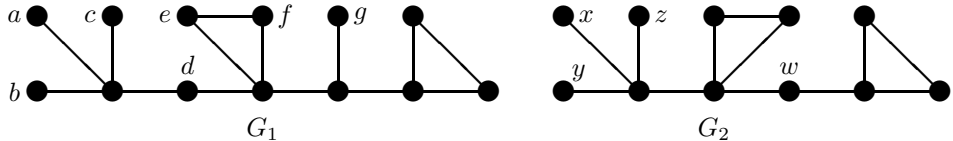
\begin{figure}[h]
\setlength{\unitlength}{1cm}\begin{picture}(5,1.8)\thicklines
\multiput(1,0.5)(1,0){7}{\circle*{0.29}}
\multiput(1,1.5)(1,0){6}{\circle*{0.29}}
\put(1,0.5){\line(1,0){6}}
\put(1,1.5){\line(1,-1){1}}
\put(2,0.5){\line(0,1){1}}
\put(3,1.5){\line(1,-1){1}}
\put(3,1.5){\line(1,0){1}}
\put(4,0.5){\line(0,1){1}}
\put(5,0.5){\line(0,1){1}}
\put(6,1.5){\line(1,-1){1}}
\put(6,0.5){\line(0,1){1}}
\put(0.7,1.5){\makebox(0,0){$a$}}
\put(0.7,0.5){\makebox(0,0){$b$}}
\put(1.7,1.5){\makebox(0,0){$c$}}
\put(2.7,1.5){\makebox(0,0){$e$}}
\put(4.3,1.5){\makebox(0,0){$f$}}
\put(5.3,1.5){\makebox(0,0){$g$}}
\put(3,0.85){\makebox(0,0){$d$}}
\put(4,0){\makebox(0,0){$G_{1}$}}
\multiput(8,0.5)(1,0){6}{\circle*{0.29}}
\multiput(8,1.5)(1,0){5}{\circle*{0.29}}
\put(8,0.5){\line(1,0){5}}
\put(8,1.5){\line(1,-1){1}}
\put(9,0.5){\line(0,1){1}}
\put(10,0.5){\line(0,1){1}}
\put(10,0.5){\line(1,1){1}}
\put(10,1.5){\line(1,0){1}}
\put(12,0.5){\line(0,1){1}}
\put(12,1.5){\line(1,-1){1}}
\put(8.3,1.5){\makebox(0,0){$x$}}
\put(8.2,0.75){\makebox(0,0){$y$}}
\put(9.3,1.5){\makebox(0,0){$z$}}
\put(11,0.8){\makebox(0,0){$w$}}
\put(10,0){\makebox(0,0){$G_{2}$}}
\end{picture}\caption{Both $G_{1}$ and $G_{2}$\ are not K\"{o}nig-Egerv\'{a}ry
graphs.}%
\label{fig22222}%
\end{figure}

Moreover,
\[
\mathrm{\ker}(G_{1})=\left\{  a,b,c\right\}  \subset\mathrm{core}%
(G_{1})\subset\left\{  a,b,c,d,g\right\}  =\mathrm{nucleus}(G_{1}),
\]
as $\mathrm{MaxCritIndep}(G)=\left\{  \left\{  a,b,c,d,e,g\right\}  ,\left\{
a,b,c,d,f,g\right\}  \right\}  $.

In addition, notice that $\mathrm{diadem}(G_{1})\subsetneq\mathrm{corona}%
(G_{1})$.

\begin{theorem}
\label{Prop1}\cite{JarLevMan2015b} Let $\emptyset\neq\Gamma\subseteq\Omega
(G)$. If $%
{\displaystyle\bigcup}
\Gamma$ is critical, then $%
{\displaystyle\bigcap}
\Gamma$ is critical as well.
\end{theorem}

It is well-known that $\alpha(G)+\mu(G)\leq\left\vert V(G)\right\vert $ holds
for every graph $G$. Recall that if $\alpha(G)+\mu(G)=\left\vert
V(G)\right\vert $, then $G$ is a \textit{K\"{o}nig-Egerv\'{a}ry graph}
\cite{Deming1979,Sterboul1979}. For example, each bipartite graph is a
K\"{o}nig-Egerv\'{a}ry graph. Various properties of K\"{o}nig-Egerv\'{a}ry
graphs can be found in \cite{Bonomo2013,Korach2006,levm4,LevMan2013b}.

\begin{theorem}
\label{th5}\cite{Larson2011,LevMan2012b} For a graph $G$, the following
assertions are equivalent:

\emph{(i)} $G$ is a K\"{o}nig-Egerv\'{a}ry graph;

\emph{(ii)} there exists some maximum independent set which is critical;

\emph{(iii)} each of its maximum independent sets is critical.
\end{theorem}

If $\Gamma,\Gamma^{\prime}$ are two collections of sets, we write
$\Gamma^{\prime}\vartriangleleft\Gamma$ if $%
{\displaystyle\bigcup}
\Gamma^{\prime}\subseteq%
{\displaystyle\bigcup}
\Gamma$ and $%
{\displaystyle\bigcap}
\Gamma\subseteq%
{\displaystyle\bigcap}
\Gamma^{\prime}$ \cite{JarLevMan2015a}. Clearly, the relation
$\vartriangleleft$ is a preorder.

\begin{theorem}
\label{th1}\cite{JarLevMan2015a} Let $\emptyset\neq\Gamma\subseteq\Omega(G)$.

\emph{(i)} If $\Gamma^{\prime}\subseteq\mathrm{Ind}(G)$ is such that
$\Gamma^{\prime}\vartriangleleft\Gamma$, then $\left\vert
{\displaystyle\bigcap}
\Gamma^{\prime}\right\vert +\left\vert
{\displaystyle\bigcup}
\Gamma^{\prime}\right\vert \leq\left\vert
{\displaystyle\bigcap}
\Gamma\right\vert +\left\vert
{\displaystyle\bigcup}
\Gamma\right\vert $.

\emph{(ii)} $2\alpha(G)\leq\left\vert
{\displaystyle\bigcap}
\Gamma\right\vert +\left\vert
{\displaystyle\bigcup}
\Gamma\right\vert $.

\emph{(iii)} If, in addition, $G$ is a K\"{o}nig-Egerv\'{a}ry graph, then
$\left\vert
{\displaystyle\bigcap}
\Gamma\right\vert +\left\vert
{\displaystyle\bigcup}
\Gamma\right\vert =2\alpha(G)$, and, in particular, $\left\vert
\mathrm{corona}(G)\right\vert +\left\vert \mathrm{core}(G)\right\vert
=2\alpha(G)$.
\end{theorem}

Let us notice that if $S\in\mathrm{Ind}(G)$, then $G[N[S]]$ is not necessarily
a K\"{o}nig-Egerv\'{a}ry graph. For instance, consider the graph $G_{1}$ from
Figure \ref{fig22222}, and $S_{1}=\left\{  d,g\right\}  $, $S_{2}=\left\{
d,e,g\right\}  $. Then, $G_{1}[N[S_{1}]]$ is a K\"{o}nig-Egerv\'{a}ry graph,
while $G_{1}[N[S_{2}]]$ is not a K\"{o}nig-Egerv\'{a}ry graph.

\begin{theorem}
\cite{Larson2011}\label{th2} For every graph $G$, there is some $X\subseteq
V\left(  G\right)  $, such that:

\emph{(i)} $X=N[S]$ for every $S\in\mathrm{MaxCritIndep}(G)$;

\emph{(ii)} $G[X]$ is a K\"{o}nig-Egerv\'{a}ry graph.
\end{theorem}

In other words, Theorem \ref{th2}\emph{(i)} claims that $X=N[S]$ does not
depend on the choice of $S\in\mathrm{MaxCritIndep}(G)$. The \textit{critical
independence number} of $G$ is defined as $\alpha^{\prime}(G)=\max\{\left\vert
S\right\vert :S\in\mathrm{MaxCritIndep}(G)\}$ \cite{Larson2011}.

Recently, the following conjectures were validated in \cite{Short2015}.

\begin{conjecture}
\label{conj1}\cite{JarLevMan2015a} If $\left\vert \mathrm{nucleus}%
(G)\right\vert +\left\vert \mathrm{diadem}(G)\right\vert =2\alpha(G)$, then
$G$ is a K\"{o}nig-Egerv\'{a}ry graph.
\end{conjecture}

\begin{conjecture}
\label{conj2} \cite{JarLevMan2015b} If $\left\vert \mathrm{diadem}%
(G)\right\vert =\left\vert \mathrm{corona}(G)\right\vert $, then $G$ is a
K\"{o}nig-Egerv\'{a}ry graph.
\end{conjecture}

\begin{conjecture}
\label{conj3} \cite{JarLevMan2015b} $\left\vert \mathrm{\ker}(G)\right\vert
+\left\vert \mathrm{diadem}(G)\right\vert \leq2\alpha(G)$ for every graph $G$.
\end{conjecture}

In this paper we involve these findings in a more general framework, where
they appear as corollaries.

\section{Results}

\begin{lemma}
\label{Lem1}If $S\in\mathrm{MaxCritIndep}(G)$ and $X=N[S]$, then
$\mathrm{MaxCritIndep}(G)\vartriangleleft\Omega\left(  G[X]\right)  $.
\end{lemma}

\begin{proof}
By Proposition \ref{Prop2}, we get that $\alpha\left(  G[X]\right)
=\left\vert S\right\vert $. Since, in accordance with Theorem \ref{th2}%
\emph{(i)}, $X=N[A]$ for each $A\in\mathrm{MaxCritIndep}(G)$, we may conclude
that $\mathrm{MaxCritIndep}(G)\subseteq\Omega\left(  G[X]\right)  $. Hence,
$\mathrm{MaxCritIndep}(G)\vartriangleleft\Omega\left(  G[X]\right)  $.
\end{proof}

There is a graph $G$ with $\mathrm{MaxCritIndep}(G)\subsetneq\Omega\left(
G[X]\right)  $, $S\in\mathrm{MaxCritIndep}(G)$, and $X=N[S]$. For instance,
the graph $G$ from Figure \ref{fig43} has
\[
\mathrm{MaxCritIndep}(G)=\left\{  \left\{  a,b,c,d,e,g\right\}  ,\left\{
a,b,c,d,f,g\right\}  \right\}  ,X=N[\left\{  a,b,c,d,e,g\right\}  ],
\]
while $\left\{  a,b,c,d,e,k\right\}  \in\Omega\left(  G[X]\right)
-\mathrm{MaxCritIndep}(G)$.

\begin{figure}[h]
\setlength{\unitlength}{1cm}\begin{picture}(5,1)\thicklines
\multiput(4,0)(1,0){7}{\circle*{0.29}}
\multiput(4,1)(1,0){6}{\circle*{0.29}}
\put(4,0){\line(1,0){6}}
\put(4,1){\line(1,-1){1}}
\put(5,0){\line(0,1){1}}
\put(6,1){\line(1,-1){1}}
\put(6,1){\line(1,0){1}}
\put(7,0){\line(0,1){1}}
\put(8,0){\line(0,1){1}}
\put(9,1){\line(1,-1){1}}
\put(9,0){\line(0,1){1}}
\put(3.7,1){\makebox(0,0){$a$}}
\put(3.7,0){\makebox(0,0){$b$}}
\put(4.7,1){\makebox(0,0){$c$}}
\put(5.7,1){\makebox(0,0){$e$}}
\put(7.3,1){\makebox(0,0){$f$}}
\put(8.3,1){\makebox(0,0){$g$}}
\put(8.3,0.35){\makebox(0,0){$k$}}
\put(6,0.35){\makebox(0,0){$d$}}
\put(2.8,0.5){\makebox(0,0){$G$}}
\end{picture}\caption{$d(\{a,b,c,d,e,k\})=1<2=d(G)$.}%
\label{fig43}%
\end{figure}
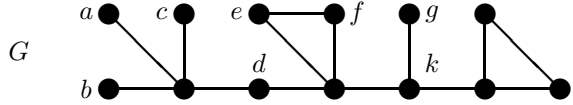

\begin{corollary}
\cite{Short2015} If $S\in\mathrm{MaxCritIndep}(G)$ and $X=N[S]$, then
\[
\mathrm{diadem}(G)\subseteq\mathrm{diadem}(G[X])\text{ and }\mathrm{nucleus}%
(G[X])\subseteq\mathrm{nucleus}(G).
\]

\end{corollary}

\begin{proof}
In accordance with Theorem \ref{th2}\emph{(ii)}, $G[X]$ is a
K\"{o}nig-Egerv\'{a}ry graph. Hence, Theorem \ref{th5}\emph{(iii)} implies
that $\mathrm{MaxCritIndep}(G[X])=\Omega\left(  G[X]\right)  $. Therefore,
Lemma \ref{Lem1} ensures that $\mathrm{MaxCritIndep}(G)\vartriangleleft
\mathrm{MaxCritIndep}(G[X])$, which, by definition, means $\mathrm{diadem}%
(G)\subseteq\mathrm{diadem}(G[X])$ and $\mathrm{nucleus}(G[X])\subseteq
\mathrm{nucleus}(G)$.
\end{proof}

\begin{theorem}
\label{th13}If $\emptyset\neq\Gamma^{\prime}\subseteq\mathrm{MaxCritIndep}(G)$
and $\emptyset\neq\Gamma\subseteq\Omega(G)$, then
\[
\left\vert
{\displaystyle\bigcap}
\Gamma^{\prime}\right\vert +\left\vert \bigcup\Gamma^{\prime}\right\vert
\leq2\alpha^{\prime}(G)\leq2\alpha(G)\leq\left\vert
{\displaystyle\bigcap}
\Gamma\right\vert +\left\vert \bigcup\Gamma\right\vert .
\]

\end{theorem}

\begin{proof}
Let $S\in\mathrm{MaxCritIndep}(G)$ and $X=N[S]$. Since $\Gamma^{\prime
}\subseteq\mathrm{MaxCritIndep}(G)$, and, by Lemma \ref{Lem1},
$\mathrm{MaxCritIndep}(G)\vartriangleleft\Omega\left(  G[X]\right)  $, we get
$\Gamma^{\prime}\vartriangleleft\Omega\left(  G[X]\right)  $. According to
Theorem \ref{th2}\emph{(ii)}, $G[X]$ is a K\"{o}nig-Egerv\'{a}ry graph. Now,
using Theorem \ref{th1}\emph{(ii), (iii)}, we obtain
\begin{gather*}
\left\vert
{\displaystyle\bigcap}
\Gamma^{\prime}\right\vert +\left\vert \bigcup\Gamma^{\prime}\right\vert
\leq\left\vert \mathrm{core}(G[X])\right\vert +\left\vert \mathrm{corona}%
(G[X])\right\vert \\
=2\alpha(G[X])=2\alpha^{\prime}(G)\leq2\alpha(G)\leq\left\vert
{\displaystyle\bigcap}
\Gamma\right\vert +\left\vert \bigcup\Gamma\right\vert ,
\end{gather*}
as claimed.
\end{proof}

If $\Gamma^{\prime}=\mathrm{MaxCritIndep}(G)$ and $\Gamma=\Omega(G)$, Theorem
\ref{th13} immediately implies the following.

\begin{corollary}
\label{corollary 8}\cite{Short2015} $\left\vert \mathrm{nucleus}(G)\right\vert
+\left\vert \mathrm{diadem}(G)\right\vert \leq2\alpha\left(  G\right)  $ for
every graph $G$.
\end{corollary}

Since $\mathrm{\ker}(G)\subseteq\mathrm{nucleus}(G)$, Corollary
\ref{corollary 8} validates Conjecture \ref{conj3}.

Let us recall that a family of independent sets $\Gamma$ is a
\textit{K\"{o}nig-Egerv\'{a}ry collection} if
\[
\left\vert
{\displaystyle\bigcap}
\Gamma\right\vert +\left\vert
{\displaystyle\bigcup}
\Gamma\right\vert =2\alpha(G)\text{ \cite{JarLevMan2015a}}.
\]
If there exists a K\"{o}nig-Egerv\'{a}ry collection $\Gamma\subseteq\Omega
(G)$, this does not oblige $G$ to be a K\"{o}nig-Egerv\'{a}ry graph. For
instance, the graph $G$ from Figure \ref{fig111}\ satisfies $\left\vert
\mathrm{corona}(G)\right\vert +\left\vert \mathrm{core}(G)\right\vert
=2\alpha(G)$, i.e., $\Omega(G)$ is a K\"{o}nig-Egerv\'{a}ry collection, while
$G$ is not a K\"{o}nig-Egerv\'{a}ry graph.\begin{figure}[h]
\setlength{\unitlength}{1cm}\begin{picture}(5,1.7)\thicklines
\multiput(4,0.5)(1,0){5}{\circle*{0.29}}
\multiput(4,1.5)(2,0){2}{\circle*{0.29}}
\put(7,1.5){\circle*{0.29}}
\put(8,1.5){\circle*{0.29}}
\put(4,0.5){\line(1,0){4}}
\put(4,0.5){\line(0,1){1}}
\put(4,1.5){\line(1,-1){1}}
\put(4,1.5){\line(2,-1){2}}
\put(4,0.5){\line(2,1){2}}
\put(5,0.5){\line(1,1){1}}
\put(6,0.5){\line(0,1){1}}
\put(7,0.5){\line(0,1){1}}
\put(8,0.5){\line(0,1){1}}
\qbezier(4,0.5)(5,-0.3)(6,0.5)
\put(3.7,1.5){\makebox(0,0){$a$}}
\put(6.3,1.5){\makebox(0,0){$b$}}
\put(7.3,1.5){\makebox(0,0){$c$}}
\put(7.3,0.68){\makebox(0,0){$d$}}
\put(8.3,1.5){\makebox(0,0){$e$}}
\put(8.3,0.5){\makebox(0,0){$f$}}
\put(3,1){\makebox(0,0){$G$}}
\end{picture}\caption{$\mathrm{core}(G)=\left\{  a,b\right\}  $ and
$\mathrm{corona}(G)=\left\{  a,b,c,d,e,f\right\}  $.}%
\label{fig111}%
\end{figure}
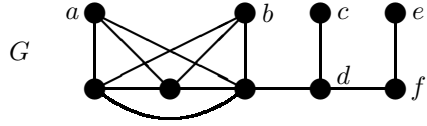

\begin{theorem}
\label{th12}For a graph $G$, the following assertions are equivalent:

\emph{(i)} $G$ is a K\"{o}nig-Egerv\'{a}ry graph;

\emph{(ii)} every non-empty family of maximum critical independent sets of $G$
is a K\"{o}nig-Egerv\'{a}ry collection;

\emph{(iii)} there is a K\"{o}nig-Egerv\'{a}ry collection of maximum critical
independent sets of $G$.
\end{theorem}

\begin{proof}
\emph{(i) }$\Longrightarrow$ \emph{(ii) }By Theorem \ref{th5}, we obtain
$\mathrm{MaxCritIndep}(G)=\Omega\left(  G\right)  $. Further, in accordance
with Theorem \ref{th1}\emph{(iii)}, each $\Gamma^{\prime}\subseteq
\mathrm{MaxCritIndep}(G)$ is a K\"{o}nig-Egerv\'{a}ry collection.

\emph{(ii) }$\Longrightarrow$ \emph{(iii) }Clear.

\emph{(iii) }$\Longrightarrow$ \emph{(i) }Let\emph{ }$\Gamma^{\prime}%
\subseteq\mathrm{MaxCritIndep}(G)$ be a K\"{o}nig-Egerv\'{a}ry collection,
$S\in\Gamma^{\prime}$ and $X=N[S]$. Since, by Lemma \ref{Lem1},\emph{
}$\mathrm{MaxCritIndep}(G)\vartriangleleft\Omega\left(  G[X]\right)  $, we
arrive at the conclusion that $\Gamma^{\prime}\vartriangleleft\Omega\left(
G[X]\right)  $, and hence,
\[
\left\vert
{\displaystyle\bigcap}
\Gamma^{\prime}\right\vert +\left\vert
{\displaystyle\bigcup}
\Gamma^{\prime}\right\vert \leq\left\vert \mathrm{nucleus}(G[X])\right\vert
+\left\vert \mathrm{diadem}(G[X])\right\vert .
\]

According to Theorem \ref{th2}\emph{(ii)}, $G[X]$ is a K\"{o}nig-Egerv\'{a}ry
graph. Using Theorem \ref{th1}\emph{(iii)}, we infer that
\[
2\alpha(G)=\left\vert
{\displaystyle\bigcap}
\Gamma^{\prime}\right\vert +\left\vert \bigcup\Gamma^{\prime}\right\vert
\leq\left\vert \mathrm{nucleus}(G[X])\right\vert +\left\vert \mathrm{diadem}%
(G[X])\right\vert =2\alpha(G[X])\leq2\alpha(G).
\]
Consequently, we obtain $\alpha(G[X])=\alpha(G)$, which ensures that $G$ is a
K\"{o}nig-Egerv\'{a}ry graph.
\end{proof}

Since $\left\vert \mathrm{nucleus}(G)\right\vert +\left\vert \mathrm{diadem}%
(G)\right\vert =2\alpha(G)$ means that $\mathrm{MaxCritIndep}(G)$ is a
K\"{o}nig-Egerv\'{a}ry collection, Theorem \ref{th12} immediately implies the following.

\begin{corollary}
\label{corollary 7}\cite{Short2015} If $\left\vert \mathrm{nucleus}%
(G)\right\vert +\left\vert \mathrm{diadem}(G)\right\vert =2\alpha(G)$, then
$G$ is a K\"{o}nig-Egerv\'{a}ry graph, which validates Conjecture \ref{conj1}.
\end{corollary}

If $\emptyset\neq\Gamma\subseteq\Omega(G)$, then none of $%
{\displaystyle\bigcap}
\Gamma$ and $%
{\displaystyle\bigcup}
\Gamma$ is necessarily critical. For instance, consider the graph $G$ from
Figure \ref{fig111}, and $\Gamma=\left\{  \left\{  a,b,c,e\right\}  ,\left\{
a,b,c,f\right\}  \right\}  \subseteq\Omega(G)$.

\begin{theorem}
\label{th15}Let $\Gamma\subseteq\Omega(G)$ and $\emptyset\neq\Gamma^{\prime
}\subseteq\mathrm{MaxCritIndep}(G)$, be such that for every $A\in
\Gamma^{\prime}$ there exists $S\in\Gamma$ enjoying $A\subseteq S$. If $%
{\displaystyle\bigcap}
\Gamma$ is a critical set, then the following assertions are true:

\emph{(i)} $%
{\displaystyle\bigcap}
\Gamma\subseteq%
{\displaystyle\bigcap}
\Gamma^{\prime}$;

\emph{(ii)} $\Gamma^{\prime}\vartriangleleft$ $\Gamma$;

\emph{(iii)} $\left\vert
{\displaystyle\bigcap}
\Gamma^{\prime}\right\vert +\left\vert
{\displaystyle\bigcup}
\Gamma^{\prime}\right\vert \leq\left\vert
{\displaystyle\bigcap}
\Gamma\right\vert +\left\vert
{\displaystyle\bigcup}
\Gamma\right\vert $;

\emph{(iv)} $%
{\displaystyle\bigcap}
\Gamma^{\prime}=%
{\displaystyle\bigcap}
\Gamma$, if, in addition, $%
{\displaystyle\bigcup}
\Gamma^{\prime}=%
{\displaystyle\bigcup}
\Gamma$.
\end{theorem}

\begin{proof}
\emph{(i)} Let $A\in\Gamma^{\prime}$ and $S\in\Gamma$, such that $A\subseteq
S$. Since $%
{\displaystyle\bigcap}
\Gamma\subseteq S$, it follows that $A\cup%
{\displaystyle\bigcap}
\Gamma\subseteq S$, and hence, $A\cup%
{\displaystyle\bigcap}
\Gamma$ is independent. By Theorem \ref{th4}, we get that $A\cup%
{\displaystyle\bigcap}
\Gamma$ is a critical independent set. Since $A\subseteq A\cup%
{\displaystyle\bigcap}
\Gamma$ and $A$ is a maximum critical independent set, we infer that $%
{\displaystyle\bigcap}
\Gamma\subseteq A$. Thus, $%
{\displaystyle\bigcap}
\Gamma\subseteq A$ for every $A\in\Gamma^{\prime}$. Therefore, $%
{\displaystyle\bigcap}
\Gamma\subseteq%
{\displaystyle\bigcap}
\Gamma^{\prime}$.

\emph{(ii)} By Part \emph{(i)}, we know that $%
{\displaystyle\bigcap}
\Gamma\subseteq%
{\displaystyle\bigcap}
\Gamma^{\prime}$. According the hypothesis, every element of $\Gamma^{\prime}$
is included in some element of $\Gamma$. Hence, we deduce that $%
{\displaystyle\bigcup}
\Gamma^{\prime}\subseteq%
{\displaystyle\bigcup}
\Gamma.$

\emph{(iii)} The inequality follows from Part \emph{(ii)} and Theorem
\ref{th1}\emph{(i)}.

\emph{(iv)} Part \emph{(iii)} implies $\left\vert
{\displaystyle\bigcap}
\Gamma^{\prime}\right\vert \leq\left\vert
{\displaystyle\bigcap}
\Gamma\right\vert $, and using Part \emph{(i)}, we obtain $%
{\displaystyle\bigcap}
\Gamma=%
{\displaystyle\bigcap}
\Gamma^{\prime}$.
\end{proof}

\begin{theorem}
\label{th14}Let $\Gamma\subseteq\Omega(G)$ and $\emptyset\neq\Gamma^{\prime
}\subseteq\mathrm{MaxCritIndep}(G)$ be such that for every $A\in\Gamma
^{\prime}$ there exists $S\in\Gamma$ enjoying $A\subseteq S$. If $%
{\displaystyle\bigcup}
\Gamma^{\prime}=%
{\displaystyle\bigcup}
\Gamma$, then $G$ is a K\"{o}nig-Egerv\'{a}ry graph.
\end{theorem}

\begin{proof}
Since, by Theorem \ref{th4}\emph{(ii)}, $%
{\displaystyle\bigcup}
\Gamma^{\prime}$ is critical, we get that $%
{\displaystyle\bigcup}
\Gamma$ is critical. Hence, according to Theorem \ref{Prop1}, we infer that $%
{\displaystyle\bigcap}
\Gamma$ is critical. Applying Theorem \ref{th15}, we obtain $%
{\displaystyle\bigcap}
\Gamma=%
{\displaystyle\bigcap}
\Gamma^{\prime}$. Further, we have%
\begin{gather*}
2\alpha(G)\leq\left\vert
{\displaystyle\bigcap}
\Gamma\right\vert +\left\vert
{\displaystyle\bigcup}
\Gamma\right\vert =\left\vert
{\displaystyle\bigcap}
\Gamma^{\prime}\right\vert +\left\vert
{\displaystyle\bigcup}
\Gamma^{\prime}\right\vert \\
\leq\left\vert \mathrm{core}(G[X])\right\vert +\left\vert \mathrm{corona}%
(G[X])\right\vert =2\alpha(G[X])\leq2\alpha(G)\text{.}%
\end{gather*}
Consequently, $\left\vert
{\displaystyle\bigcap}
\Gamma^{\prime}\right\vert +\left\vert
{\displaystyle\bigcup}
\Gamma^{\prime}\right\vert =2\alpha(G)$, which ensures, by Theorem \ref{th12},
that $G$ is a K\"{o}nig-Egerv\'{a}ry graph.
\end{proof}

If $\Gamma^{\prime}=\mathrm{MaxCritIndep}(G)$ and $\Gamma=\Omega(G)$, Theorem
\ref{th14} immediately implies the following.

\begin{corollary}
\cite{Short2015} If $\mathrm{diadem}(G)=\mathrm{corona}(G)$, then $G$ is a
K\"{o}nig-Egerv\'{a}ry graph, which validates Conjecture \ref{conj2}.
\end{corollary}

\section{Conclusions}

In this paper we focus on interconnections between unions and intersections of
maximum critical independents sets of a graph.

In \cite{Short2015}\ the question about possible polynomial complexity of the
lower bounds
\[
\left\vert \mathrm{nucleus}(G)\right\vert +\left\vert \mathrm{diadem}%
(G)\right\vert \leq2\alpha\left(  G\right)
\]
for every graph $G$ arises.

Let $S\in\mathrm{MaxCritIndep}(G)$.

\begin{proposition}
\label{Prop3}$2\alpha^{\prime}(G)=d\left(  \mathrm{\ker}(G)\right)
+\left\vert N[S]\right\vert $.
\end{proposition}

\begin{proof}
Since $\mathrm{\ker}(G)$ and $S$ are critical sets of the graph $G$, we
obtain
\begin{gather*}
d\left(  \mathrm{\ker}(G)\right)  +\left\vert N[S]\right\vert =\left\vert
\mathrm{\ker}(G)\right\vert -\left\vert N\left(  \mathrm{\ker}(G)\right)
\right\vert +\left\vert S\right\vert +\left\vert N\left(  S\right)
\right\vert \\
=\left\vert S\right\vert -\left\vert N\left(  S\right)  \right\vert
+\left\vert S\right\vert +\left\vert N\left(  S\right)  \right\vert
=2\left\vert S\right\vert =2\alpha^{\prime}(G),
\end{gather*}
which completes the proof.
\end{proof}

By Theorem \ref{th5} and Proposition \ref{Prop3}, if $G$\ is a
K\"{o}nig-Egerv\'{a}ry graph, then we get
\[
2\alpha(G)=d\left(  \mathrm{\ker}(G)\right)  +\left\vert N[S]\right\vert ,
\]
because $\alpha^{\prime}(G)=\alpha(G)$. Consequently,
\[
2\alpha\left(  G\right)  \leq\left\vert \mathrm{\ker}(G)\right\vert
+\left\vert N\left[  S\right]  \right\vert ,
\]
which motivates the following.

\begin{problem}
Characterize graphs such that
\[
2\alpha\left(  G\right)  \leq\left\vert \mathrm{\ker}(G)\right\vert
+\left\vert N\left[  S\right]  \right\vert \leq\left\vert \mathrm{\ker
}(G)\right\vert +2\alpha^{\prime}(G)\leq3\alpha^{\prime}(G).
\]
Let us call them approximate K\"{o}nig-Egerv\'{a}ry graphs.
\end{problem}

By Proposition \ref{Prop3},
\[
\left\vert \mathrm{\ker}(G)\right\vert +\left\vert N\left[  S\right]
\right\vert -\left\vert N\left(  \mathrm{\ker}(G)\right)  \right\vert
=2\alpha^{\prime}(G)\leq2\alpha\left(  G\right)  .
\]
Since $\mathrm{\ker}(G)$ \cite{LevMan2012a}, $\left\vert N\left[  S\right]
\right\vert $ and $\alpha^{\prime}(G)$ \cite{Larson2007} can be computed
polinomially, it means that for approximate K\"{o}nig-Egerv\'{a}ry graphs
their independence numbers are bounded as follows%
\begin{gather*}
\frac{\left\vert \mathrm{\ker}(G)\right\vert +\left\vert N\left[  S\right]
\right\vert }{2}-\frac{\left\vert N\left(  \mathrm{\ker}(G)\right)
\right\vert }{2}=\alpha^{\prime}(G)\leq\alpha\left(  G\right) \\
\leq\frac{\left\vert \mathrm{\ker}(G)\right\vert +\left\vert N\left[
S\right]  \right\vert }{2}\leq\frac{\left\vert \mathrm{\ker}(G)\right\vert
}{2}+\alpha^{\prime}(G)\leq\frac{3}{2}\alpha^{\prime}(G)\text{.}%
\end{gather*}

\end{document}